\documentclass{ws-mpla}
\usepackage[super]{cite}
\usepackage{graphicx}

\begin{document}

\markboth{B. S. Ratanpal, B. Suthar and V. Shah}{Solution of EFEs for Anisotropic Matter with Vanishing Complexity}

\catchline{}{}{}{}{}

\title{Solution of Einstein Field Equations for Anisotropic Matter with Vanishing Complexity: Spacetime Metric Satisfying Karmarkar Condition and Conformally Flat Geometry}

\author{B. S. Ratanpal}

\address{Department of Applied Mathematics, Faculty of Technology \& Engineering, The Maharaja Sayajirao University of Baroda, Vadodara, Gujarat 390 001, India\\
bharatratanpal@gmail.com}

\author{Bhavesh Suthar\footnote{Corresponding Author}}

\address{Department of Applied Mathematics, Faculty of Technology \& Engineering, The Maharaja Sayajirao University of Baroda, Vadodara, Gujarat 390 001, India\\
bhaveshsuthar.math@gmail.com}

\author{Vishant Shah}

\address{Department of Applied Mathematics, Faculty of Technology \& Engineering, The Maharaja Sayajirao University of Baroda, Vadodara, Gujarat 390 001, India\\
vishantmsu83@gmail.com}

\maketitle

\pub{Received (Day Month Year)}{Revised (Day Month Year)}

\begin{abstract}
The solution of Einstein's field equations for static spherically symmetric spacetime metric with anisotropic internal stresses has been obtained. The matter has vanishing complexity and a spacetime metric that satisfies the Karmarkar condition and is conformally flat. It has been noted that there is only one solution that meets these three conditions. This has been shown as a proof of the theorem.

\end{abstract}

\ccode{PACS Nos.:}

\section{Introduction}	

In 2018, Herrera \cite{Herrera2018} introduced a novel definition of complexity for the static spherically symmetric self-gravitating fluid distribution. Andrade and Andrade \cite{Andrade2024} examined the complexity factor in gravitational decoupling with minimal geometric deformation. Sharif and Tariq \cite{Sharif2020} have investigated the solution of Einstein-Maxwell field equations for the evaluation of self-gravitating system. The static anisotropic stellar models with vanishing complexity have been studied by Contreras \textit{et. al.} \cite{Contreras2022complex} , Bogadi \textit{et. al.} \cite{Bogadi} , Rej \textit{et. al.} \cite{Rej2024} , Das \textit{et. al.} \cite{Das2024a} and Das \textit{et. al.} \cite{Das2024b}. \\

\noindent The embedding problem is of geometrical significance, which was first discussed by Schlai \cite{Schlai} . The first isometric embedding theorem was proposed by Nash \cite{Nash} . In order to embed a 4-dimensional spacetime metric in a 5-dimensional Euclidean space, Karmarkar \cite{Karmarkar1948} derived the condition that includes curvature components. This spacetime metric is said to be class-I. The models satisfying the Karmarkar condition attracted the attention of many researchers, viz., Maurya \textit{et. al.} \cite{Maurya2015} , Bhar \textit{et. al.} \cite{Bhar2016} , Maurya \textit{et. al.} \cite{Maurya2017} , Francisco \textit{et. al.} \cite{Francisco} , Ratanpal \textit{et. al.} \cite{Ratanpal2023}.\\

\noindent If all components of the Weyl tensor are zero, the spacetime metric is regarded as conformally flat. The conformally flat spacetime can be found in the work of Herrera \textit{et. al.} \cite{Herrera2014} , Manjonjo \textit{et. al.}\cite{Manjonjo2017} and Manjonjo \textit{et. al.}\cite{Manjonjo2018}.\\

\noindent In this work, we investigated a solution for a static spherically symmetric anisotropic fluid distribution with vanishing complexity and a spacetime metric that satisfies the Karmarkar condition and is conformally flat. The work is organised as follows: Section 2 contains a description of the complexity factor, the Karmarkar condition, and the conformally flat spacetime metric. Section 3 describes the solution for interior static spherically symmetric stellar structure with vanishing complexity as well as spacetime metric satisfying Karmarkar condition and is conformally flat. We concluded the work in Section 4.

\section{Einstein field equations}

For a static spherically symmetric fluid distribution, the interior spacetime metric is described as
\begin{equation}\label{metric1}
    ds^{2}=e^{\nu\left(r\right)}dt^{2}-e^{\lambda\left(r\right)}dr^{2}-r^{2}\left({d\theta}^{2}+\sin^{2}{\theta} {d\phi}^{2}\right).
\end{equation}
The energy-momentum tensor for the anisotropic fluid distribution is defined as
\begin{equation} \label{EMTensor}
    T_{ij}=\left(\rho+p\right)u_{i}u_{j}-p g_{ij} + \pi_{ij},         \hspace{1cm} u^{i}u_{i}=1,
\end{equation}
where $\rho$ and $p$ indicate the density and isotropic pressure, respectively; $u^{i}$ is the unit four velocity; and $\pi_{ij}$ is the anisotropic stress tensor, which is defined as
\begin{equation} \label{piij}
    \pi_{ij}=\sqrt{3} S \left[c_{i} c_{j} - \frac{1}{3} \left(u_{i} u_{j}-g_{ij}\right)\right],
\end{equation}
where $S=S\left(r\right)$ represents the magnitude of anisotropy and $c^{i}=\left(0,e^{-\frac{\lambda}{2}},0,0\right)$ represents the radially directed vector.\\

\noindent According to Herrera \cite{Herrera2018} , the complexity factor for static self-gravitating fluids with spherical symmetry could be one of the structural scalars that result from the orthogonal decomposition of the Riemann tensor. The complexity factor for a static configuration is defined as
\begin{equation} \label{CF4}
    Y_{TF}=8\pi \Pi-\frac{4\pi}{r^{3}} \int_{0}^{r} \left[\rho\left(\Tilde{r}\right)\right]^{'} {\Tilde{r}}^{3} d\Tilde{r},
\end{equation}
using the spacetime metric (\ref{metric1}) and $Y_{TF}$ (\ref{CF4}), we may describe complexity as
\begin{equation} \label{CF6}
    Y_{TF}=\frac{e^{-\lambda}\left[\nu^{'}\left(r \lambda^{'}-r \nu^{'}+2\right)-2 r \nu^{''} \right]}{4r}.
\end{equation}
Therefore, the vanishing complexity entails
\begin{equation}
    Y_{TF}=0,
\end{equation}
\begin{equation} \label{VCF7}
\implies
    \nu^{'}\left(r\lambda^{'}-r\nu^{'}+2\right)-2r\nu^{''}=0.
\end{equation}
This is a second-order nonlinear differential equation. Contreras and Stuchlik \cite{Contreras2022minimal} provide the following solution to the above equation given as
\begin{equation} \label{enu8}
    e^{\nu}=\left[A+B\int r e^{\frac{\lambda}{2}}dr\right]^{2}.
\end{equation}
The spacetime metric (\ref{metric1}) is generally considered to be of class-II. Karmarkar \cite{Karmarkar1948} provided a condition for spacetime metric (1) to be of class-I, which means that a 4-dimensional spacetime metric can be embedded in 5-dimensional Euclidean space if
\begin{equation} \label{Karmarkar9}
    R_{1414}=\frac{R_{1212} R_{3434}+R_{1224} R_{1334}}{R_{2323}},
\end{equation}
where $R_{2323}\neq 0$.
For spacetime metric (\ref{metric1}), the components of the Riemann curvature tensor are as follows:
\begin{eqnarray}
R_{1414} &=& -e^{\nu}\left(\frac{\nu^{''}}{2}+\frac{{\nu^{'}}^{2}}{4}-\frac{\lambda^{'}\nu^{'}}{4}\right),\\
R_{2323} &=& -e^{-\lambda} r^{2} \sin^{2}\theta \left(e^{\lambda}-1\right),\\
R_{1212} &=& \frac{r \lambda^{'}}{2},\\
R_{3434} &=& -\frac{1}{2} r \sin^{2} \theta \nu^{'} e^{\nu-\lambda}.
\end{eqnarray}
Substituting these Riemann curvature tensors in the Karmarkar condition (\ref{Karmarkar9}), we obtain
\begin{equation} \label{nu15}
    2 \nu^{''}+{\nu^{'}}^{2}=\frac{\nu^{'}\lambda^{'}e^{\lambda}}{e^{\lambda}-1}.
\end{equation}
Integrating (\ref{nu15}) the solution of the field equation is given by
\begin{equation} \label{enu16}
    e^{\nu}=\left[C+D \int \sqrt{e^{\lambda}-1} dr\right]^{2},
\end{equation}
where C and D are constants of integration.
Weyl described the Weyl tensor in n-dimensional Riemannian space. For $n=2, 3,$ the Weyl tensor is zero. Spacetime is considered conformally flat for $n\geq 4$ when all components of the Weyl tensor vanish yielding
\begin{equation} \label{conformally}
    \frac{1-e^{\lambda}}{r^{2}}-\frac{\nu^{'}\lambda^{'}}{4}-\frac{\nu^{'}-\lambda^{'}}{2r}+\frac{\nu^{''}}{2}+\frac{{\nu^{'}}^2}{4}=0,
\end{equation}
which was provided by Ponce de Leon \cite{Ponce} .

\section{Vanishing Complexity, Karmarkar Condition and Conformal Flatness}
This section provides a solution to Einstein's field equations that have vanishing complexity and a spacetime metric satisfying the Karmarkar condition as well as conformally flat.

\begin{theorem}
    (Existence and Uniqueness theorem)
    For a spherically symmetric interior spacetime metric (\ref{metric1}) with an energy-momentum tensor (\ref{EMTensor}), if the matter distribution in the interior of a stellar configuration has vanishing complexity, the spacetime metric satisfies the Karmarkar condition, then there exists a unique solution to Einstein's field equations which is conformally flat.
\end{theorem}

\begin{proof}
    The matter distribution in the interior with vanishing complexity implies
\begin{equation}
    \nu^{'}\left[r\left(\lambda^{'}-\nu^{'}\right)+2\right]-2r\nu^{''}=0,
\end{equation}
\begin{equation} \label{eqn18}
    \therefore r \nu^{'} \lambda^{'}+2 \nu^{'}-r\left(2 \nu^{''}+{\nu^{'}}^{2}\right)=0.
\end{equation}
From Karmarkar condition (\ref{nu15}),
\begin{equation}
    2 \nu^{''}+{\nu^{'}}^{2}=\frac{\nu^{'}\lambda^{'}e^{\lambda}}{e^{\lambda}-1}.
\end{equation}
Substituting above in (\ref{eqn18})
\begin{equation}
    \left(r \lambda^{'}+2\right) \nu^{'}-\frac{r \nu^{'} \lambda^{'} e^{\lambda}}{e^{\lambda}-1}=0,
\end{equation}
\begin{equation}
    \therefore r \lambda^{'}+2-\frac{r\lambda^{'}e^{\lambda}}{e^{\lambda}-1}=0,
\end{equation}
\begin{equation}
    \therefore r \left[1-\frac{e^{\lambda}}{e^{\lambda}-1}\right] \lambda^{'}+2=0,
\end{equation}
\begin{equation}
    \therefore \frac{r \lambda^{'}}{e^{\lambda}-1}=2,
\end{equation}
\begin{equation}
    \therefore \frac{d\lambda}{e^{\lambda}-1}=\frac{2}{r}dr,
\end{equation}
integrating we get,
\begin{equation} \label{elambda}
    e^{\lambda}=\frac{1}{1-\frac{r^{2}}{R^{2}}}.
\end{equation}
Using (\ref{elambda}) in (\ref{enu16}) and integrating,
\begin{equation} \label{enu27}
    \therefore e^{\nu}=\left[C-D \sqrt{1-\frac{r^{2}}{R^{2}}}\right]^{2}.
\end{equation}
Substituting (\ref{elambda}) and (\ref{enu27}), the spacetime metric (\ref{metric1}) can be expresses as
\begin{equation}
    ds^{2}=\left[C-D \sqrt{1-\frac{r^{2}}{R^{2}}}\right]^{2} dt^{2}-\frac{1}{1-\frac{r^{2}}{R^{2}}} dr^{2}-r^{2}\left(d\theta^{2}+\sin^{2}\theta d\phi^{2}\right),
\end{equation}
which is a unique solution to Einstein's field equations named Schwarzschild interior solution. This solution satisfies condition (\ref{conformally}). Hence, it is conformally flat also it is well known that Schwarzschild interior spacetime metric is conformally flat.  
\end{proof}

\noindent Hence we proved the existence and uniqueness theorem for the solution of Einstein's field equations in the interior of a stellar configuration where matter has vanishing complexity, spacetime metric satisfies the Karmarkar condition and is conformally flat.

\section{Conclusion}
Many researchers have drawn attention to vanishing complexity and the Karmarkar condition in the literature. The models of stellar configuration with vanishing complexity are available in the literature; separately, models of stellar configuration satisfying the Karmarkar condition and the conformally flat condition are available in the literature. We searched for the solution of stellar configuration with vanishing complexity, in which the spacetime metric satisfies the Karmarkar condition and is conformally flat. It is found that the Schwarzschild interior solution satisfies the vanishing complexity condition and the Karmarkar condition simultaneously, which is also conformally flat. The Schwarzschild interior solution represents the matter with uniform density. Long ago, Bowers and Liang \cite{Bowers} studied the anisotropic sphere with constant density. Dev and Gleiser \cite{Dev2002} have studied the stability of anisotropic spheres with constant density. This study offers a unique solution to Einstein's field equations in the interior of a stellar configuration with vanishing complexity, where the spacetime metric satisfies both the Karmarkar condition and the conformally flat condition.

\section*{Acknowledgment}
BSR, BS, and VS would like to express their gratitude to IUCAA, Pune, for the hospitality and facilities that were provided to them during the completion of the work.
\vspace{1in}

\bibliographystyle{ws-mpla}
\bibliography{output}
\end{document}